\documentclass[11pt,reqno]{amsart}
\usepackage[samelinetheorem,notitle]{maherart}

\usepackage{epigraph}

\usepackage{cancel}
\usepackage{mathrsfs}
\usepackage{bigints}

\usepackage{thmtools} 
\usepackage{thm-restate}

\usepackage{tikz}
\usetikzlibrary{arrows,positioning}
\usetikzlibrary{fit,shapes.geometric}
\usepackage{standalone}
\usepackage{enumitem}

\usepackage{skak}

\hypersetup{
    pdftitle =
        {Multi-Dimensional Screening with Buyer-Optimal Learning},
    pdfauthor =
        {Rahul Deb and Anne-Katrin Roesler},
    pdfsubject=
        {Multi-Dimensional Screening, Bundling, Information Design, Rational Inattention}
}

\newcommand{\oa}{\overline{\alpha}}
\newcommand{\ua}{\underline{\alpha}}
\newcommand{\ot}{\overline{\theta}}
\newcommand{\ut}{\underline{\theta}}

\newcommand{\os}{\overline{s}}

\newcommand{\td}{\ \text{d}}

\newcommand{\mM}{\mathcal{M}}
\newcommand{\umM}{\underline{\mathcal{M}}}

\begin{document}

\defcitealias{RoSzAER2017}{RS}

\title{Multi-Dimensional Screening: Buyer-Optimal Learning and Informational Robustness}

\author[Deb]{Rahul Deb$^{\between}$}
\address{$^{\between}$Department of Economics, University of Toronto\\\href{mailto:rahul.deb@utoronto.ca}{rahul.deb@utoronto.ca}}
\author[Roesler]{Anne-Katrin Roesler$^{\symking}$ \\ \today}
\address{$^{\symking}$Department of Economics, University of Toronto\\\href{mailto:ak.roesler@utoronto.ca}{ak.roesler@utoronto.ca}}

\thanks{We are grateful for comments from Gabriel Carroll, Andrea Prat, Xianwen Shi, Gabor Virag, anonymous reviewers from the ACM-EC conference, and numerous seminar participants. We thank the Social Sciences and Humanities Research Council for their continued and generous financial support.}

\begin{abstract}
A monopolist seller of multiple goods screens a buyer whose type is initially unknown to both but drawn from a commonly known distribution. The buyer privately learns about his type via a signal. We derive the seller's optimal mechanism in two different information environments. We begin by deriving the buyer-optimal outcome.  Here,  an information designer first selects a signal, and then the seller chooses an optimal mechanism in response; the designer's objective is to maximize consumer surplus. Then, we derive the optimal informationally robust mechanism. In this case, the seller first chooses the mechanism, and then nature picks the signal that minimizes the seller's profits. We derive the relation between both problems and show that the optimal mechanism in both cases takes the form of pure bundling.
\end{abstract}

\maketitle

\section{Introduction}

What is the optimal mechanism that a monopolist should use to sell multiple goods to a single buyer? Despite being a classic economic problem, multi-dimensional screening is notoriously intractable. Even if the seller has just two goods and the buyer's values are additive, independent, and identically distributed, the optimal mechanism is hard to characterize generally. In this paper, we study a general version (with arbitrarily many goods and non-additive values) of this problem but with the novel feature of buyer learning. As it turns out, introducing this new feature makes the model tractable and in certain environments---including the one with independent and additive values---makes pure bundling an optimal mechanism. 

The buyer in our model initially has an unknown type $(\theta_1,\dots,\theta_n)$ that is drawn from a commonly known exchangeable\footnote{This is an assumption of symmetry that requires every permutation of the type vector to have the same joint distribution. It does not rule out positive or negative correlations.} distribution, where each $\theta_i\in [\ut,\ot]\subset \mathbb{R}_+$. The buyer's type determines his value $\kappa_b\sum_{i\in b} \theta_i$ for any bundle $b\subseteq \{1\dots,n\}$ of goods where $\kappa_b\geq 0$ is a non-negative constant. We assume a weak free-disposal property, which requires that the value of the grand bundle (that is, the bundle of all $n$ goods) is greater than any other bundle for every type. This class of value functions allows for goods to be complements or substitutes and, importantly, includes additive values ($\kappa_b=1$ for all bundles $b$) as a special case. The buyer learns about his type via a signal. Upon privately observing the signal realization, the buyer forms a posterior estimate of his value for different bundles. 

\medskip

Our aim is to derive the seller's optimal mechanism under two different informational environments. We first characterize the \textit{buyer-optimal outcome}: this is the signal and the corresponding optimal mechanism for the seller that generate the \textit{maximal consumer surplus}. Specifically,  an information designer\footnote{As we discuss below, this could be interpreted either as a theoretical benchmark or alternatively we could think of the designer as a regulator having the best interest of consumers in mind.} first publicly picks the signal (the signal realization remains private to the buyer) to maximize consumer surplus anticipating that the seller will choose an optimal mechanism in response. We show that the buyer-optimal outcome is generated by a signal that makes ``pure bundling'' (selling the grand bundle at a given price) an optimal mechanism for the seller. Additionally, we show that the seller's profit is minimized: there is no other signal and corresponding optimal mechanism that yield a lower profit. Thus, we show that the seller's profit in the buyer-optimal outcome is the solution to a min-max problem where an adversarial nature picks a signal with the aim of minimizing the profits of a seller who best-responds with an optimal mechanism.

\medskip

We then derive the \textit{optimal informationally robust mechanism} for the seller: this is an optimal mechanism for a seller who does not know how the buyer learns and who evaluates profits according to a worst-case criterion. Here, the timing is reversed: the seller first chooses the mechanism, following which nature picks the signal to minimize the seller's profit. Therefore, in this case, the seller's profit from the optimal informationally robust mechanism is the solution to a max-min problem. Once again, we show that pure bundling is optimal for the seller, but, in this case, she randomizes over the price for the grand bundle. Moreover, we derive this result by showing that the seller's profit, in this case, is exactly equal to her profit from the buyer-optimal outcome; that is, the optimal value of the objective function in the max-min and the min-max problems coincide.

\medskip

In our view, the solutions to both problems are individually economically interesting and have distinct implications. At a high level though, both demonstrate different important properties of pure bundling. The buyer-optimal outcome is a natural theoretical benchmark. The seller of a single good always finds it optimal to screen by simply posting a price. By contrast, optimal multidimensional screening can, and frequently does, involve complex menus and randomization even when values are additive and each $\theta_i$ is independently and identically distributed. Such elaborate screening helps sellers maximize profits, but the effect on the consumer is unclear. For instance, complex screening might lead to Pareto improvements where both the seller and the buyer are better off because the efficiency of trade increases. The buyer-optimal outcome is a natural benchmark to study the tradeoff between mechanism complexity and the efficiency of trade because the seller best responds to the most advantageous information structure for the buyer. Here, the optimal mechanism takes the very simple form of pure bundling, and we show that trade is efficient.

The selling practices of multi-product retailers are scrutinized by regulators who specifically express concerns about and pursue litigation against practices like tying and bundling by large firms.\footnote{Often, the concern is that such practices are used to prevent the entry of competitors. Importantly, however, a selling practice is not illegal if it maximizes profits even if it deters entry as a side effect. Instead, regulators typically scrutinize firms that incur deliberate losses to maintain market power.} Under the buyer-optimal signal, not only does pure bundling not cause consumer harm, it leads to the highest possible consumer surplus and efficient trade. This suggests that the information available to buyers is an important factor that should determine whether or not bundling needs to be scrutinized. Of course, this also raises the question of whether and, if so, which advertising practices should be regulated in the interest of consumers (see \citealp{bagwell2007} and \citealp{woodcock2017} for a recent case in favor of greater regulation).

Conversely, the optimal informationally robust mechanism provides a positive explanation for why we should expect to observe pure bundling in practice. Despite having historical data from different markets, sellers are unlikely to have very precise estimates of a buyer's value distribution. In particular, it is impossible for a seller to predict what information the buyer has or will acquire in any particular period. Our results show that pure bundling (albeit with a random price) provides the highest revenue guarantee. This is perhaps one reason why, in practice, we do not observe very complex screening that depends on fine details of the type distribution (as is possible in multidimensional screening). Instead, pure bundling is the common way that digital goods such as streaming services are sold. This setting is a good fit for the model: sellers such as Netflix and Spotify have considerable market power, and disposal is free.

\medskip

Before moving on to the related literature, we provide some high-level intuition for the main insight driving our results. The tractability of our model, given its generality and the simple form of the seller's optimal mechanisms might come across as surprising to some readers. After all, we add buyer learning to the already complex (and unsolved at this level of generality) screening problem. Learning itself can also be significantly more complex than its one-dimensional counterpart. This is because, in this multidimensional environment, the set of possible signals is very rich and does not have a simple characterization akin to the one-dimensional case (where the prior type distribution is a mean-preserving spread of the distribution of posterior estimates that can arise from any signal). Our central insight is that it is precisely the buyer learning that makes the problem tractable. Specifically, we identify a signal that describes the buyer's information in the buyer-optimal outcome. For \textit{this} particular signal, it is possible to derive the seller optimal mechanism. Moreover, this signal has the property that it minimizes the seller's revenue when she chooses the optimal informationally robust mechanism and plays an important role in its derivation.

An important feature of our environment is that generically there are informative signals such that the posterior estimates of each $\theta_i$ have maximal positive correlation across $i$. This can be done even when the prior distribution is such that different dimensions of the type vector are distributed independently or are negatively correlated. As an example, suppose the buyer observed a signal (about the sum of his type vector) of the form $\theta_1+\cdots+\theta_n +\varepsilon$ where $\varepsilon$ is some noise. Maximal positive correlation in the distribution of the posterior estimate of the type vector then follows from the fact that the prior distribution is exchangeable. 

In addition to generating such correlation, the signal in the buyer-optimal outcome adds further noise that determines the shape of the distribution of the posterior estimate of each $\theta_i$. It might seem counterintuitive to some that the signal in the buyer-optimal outcome contains so much noise. After all, imperfect information about his type affects the buyer's ability to make good purchasing choices (in an ex-post sense) when faced with a given mechanism. However, the signal also determines the seller's optimal mechanism and, therefore, the share of surplus that the seller can appropriate. Specifically, the correlation in the signal limits the seller's ability to screen across dimensions. Further, the induced shape of the distribution of the posterior type estimate results in efficiency and gives the buyer the maximal surplus share possible (by making it optimal for the seller to offer a pure bundling mechanism at a price that is favorable for the buyer). From a technical perspective, we show that such signals allow us to map our multi-dimensional screening problem onto its one-dimensional counterpart, which in turn allows us to leverage existing results from the literature.

\subsection*{Related Literature}

This paper lies at the intersection of two different literatures. The first literature examines the classic question of how a monopolist should jointly sell multiple goods. Despite being a mature literature (dating back till at least \citet{adams1976}), the complexity of the problem is such that there are surprisingly few general insights even for the special case of additive values. A seminal result is due to \citet{mcafee1989} who show that when values for each good are independent, selling goods individually (separate sales) is never optimal for the monopolist. 

In general, the optimal mechanism can be extremely complex even when values are independent and additive. \citet{pavlov2011} shows that optimal screening can involve randomization when values are identically and uniformly distributed. \citet{DaDeTzEcma2017} show that the optimal mechanism for two goods features an infinite menu of lotteries when the values are drawn independently from the beta distribution. In fact, the seller might get a negligible fraction of the optimal revenue if she is restricted to using ``simple mechanisms'' like pure bundling or separate sales \citep{hart2019}. Unlike these results, we show that, even under general exchangeable prior type distributions, simple pure bundling is optimal in two different environments with buyer learning.

Perhaps the closest paper in this literature is \citet{carroll2017}. He considers a seller who knows the marginal distribution of the buyer's value for each good but not the joint distribution. The monopolist chooses a mechanism that maximizes the worst-case profit computed over all joint distributions which have the given marginals. He shows that separate sales is seller optimal for this criterion. In addition to being a distinct economic problem, our setting is fundamentally different: the buyer's posterior value distribution must be obtained by Bayesian updating from the signal. In particular, we cannot introduce arbitrary correlations across goods, and the marginal posterior distribution for the value of each good can differ from the prior.

Additionally, unlike \citet{carroll2017}, we allow for non-additive values. There is substantially less work that analyzes these environments. The recent survey of \citet{armstrong2016} describes a strand of this literature that does not aim to derive the optimal mechanism but instead characterize conditions under which the seller can profit from offering bundle discounts. A notable exception is \citet{HaHaRES2020} who characterize environments where pure bundling is seller optimal, and we employ one of their results in the proof of our main result. Specifically, we show that, while pure bundling is not an optimal mechanism for the prior type distribution, it emerges as optimal due to buyer learning.

\medskip

This paper is also related to the growing literature on information design: \citet{BeMoJEL2019} and \citet{KaARE2019} are recent surveys. Within this literature, we are most closely related to the recent work studying how the information environment affects the selling mechanism, efficiency, and the resulting surplus division in bilateral trade settings.\footnote{Similar ideas can also be found in the literature on information acquisition and disclosure in mechanism design settings such as \citet{PeEcma2000}, \citet{BeVaEcma2002} or \citet{ShiGEB2012}. In contrast to the information design literature, these papers usually consider a restricted domain of feasible information structures.} While most of this work considers different information environments to this paper, the key distinction is that we study multi-dimensional screening. \citet{BBMAER2015} study a standard single-good monopoly pricing problem and analyze which buyer-seller surplus pairs are achievable when the seller (instead of the buyer) receives additional information which she can use to price discriminate. \citet{RaRoSzWP2020} study an environment where buyer-learning is unobservable but costly. Their main result shows that there is a distinction between free and arbitrarily cheap learning. 

The closest related papers are \citet{RoSzAER2017} and \citet{DuEcma2018}; they respectively analyze the one-dimensional versions of the two problems we study. \citet{RoSzAER2017} derive the buyer-optimal outcome for a single good, and their main insight is to show that, even if information is free, the buyer prefers not to perfectly learn his value for the product. \citet{DuEcma2018} derives the optimal informationally robust mechanism (a random posted price) for a single good\footnote{He additionally constructs an informationally robust auction to sell a common-value good which has the property that, as the number of bidders gets large, its revenue guarantee converges to the full surplus.} and uncovers the relation to the buyer-optimal outcome. The richness of the multi-dimensional screening environment that we consider opens the door to questions that cannot be addressed in the one-dimensional context. Namely, our main contribution is to derive the \textit{qualitative} properties of the seller's optimal mechanism (it takes the form of pure bundling) in both information environments; for the seller of a single good, it is always just a posted price (either deterministic or random).

\section{The Model}\label{sec:model}

We consider a mechanism design problem with one buyer and one seller, the latter of whom has one unit of each of $n\geq 2$ goods for sale. We denote the set of goods by $N=\{1, \dots, n\}$. We assume that the seller's cost of producing these goods is 0.

\medskip

\noindent \textbf{Type Space:} The buyer has a \textit{type} $\theta=(\theta_1,\dots,\theta_n)$ that lies in a set $\Theta=\left[\theta_{\ell},\theta_h\right]^n$ (endowed with the Borel $\sigma$-algebra $\mathscr{F}$) with $\theta_h>\theta_{\ell}\geq 0$. The type is initially unknown to both the buyer and seller, and is drawn from a commonly known (cumulative) distribution $F$.\footnote{We abuse notation and interchangeably refer to $F$ as a cumulative distribution (henceforth, cdf) and a probability measure. The meaning will be clear from the argument (element vs. set).} We assume $F$ has a positive density for all $\theta\in \Theta$.

Additionally, we assume that $F$ is \textit{exchangeable}: for any permutation $\sigma:N\to N$, the joint distribution of $(\theta_1,\dots,\theta_n)$ is the same as the joint distribution of $(\theta_{\sigma(1)},\dots,\theta_{\sigma(n)})$ (both of which are $F$).\footnote{Formally, for any $X\in\mathscr{F}$, we have $F(X)=F(X_{\sigma})$ where $X_{\sigma}=\{(\theta_{\sigma(1)},\dots,\theta_{\sigma(n)})\,|\, (\theta_1,\dots,\theta_n)\in X\}$.} 

Exchangeability requires the marginal distribution of each $\theta_i$ to be the same and is clearly satisfied when each dimension of the type, $\theta_i$, is independent and identically distributed (henceforth iid). Exchangeability allows both positive and negative correlations between dimensions of the type vector, although it does imply a lower bound on the correlation coefficient. This bound is increasing in $n$ but, for the most widely examined case of two goods, exchangeability imposes no restriction on the degree of correlation.

\medskip

Given a type $\theta$, we use $\overline{\theta}\in \overline{\Theta}:= [n\theta_{\ell},n\theta_h]$ to denote the sum $\overline{\theta}:=\theta_1+\cdots+\theta_n$ and $\overline{F}$ denotes the distribution of the sum $\overline{\theta}$ induced by the type distribution $F$.

\medskip

\noindent \textbf{Value Function:} Given a type $\theta$, the buyer's value for a bundle $b\subseteq N$ is
\begin{equation}\label{eq: value function}
u(\theta,b)=\kappa_b \sum_{i\in b} \theta_i,
\end{equation}
where $\kappa_b\geq 0$ and we normalize the constant for the \textit{grand bundle} $N$ to $\kappa_N=1.$ Because of this normalization, we also refer to the sum $\overline{\theta}$ as the \textit{grand bundle value}.

We assume that the buyer's value of not receiving a good is $u(\theta,\emptyset)= 0$ and that $u(\theta,N)\geq u(\theta,b)$ for all $b\subseteq N$ and all $\theta\in \Theta$. The latter is a \textit{weak free-disposal} property and ensures that the greatest surplus is generated by trading the grand-bundle. In terms of the $\kappa_b$-s, this requires that $\kappa_b \leq 1 + \frac{N-|b|}{|b|}\frac{\theta_{\ell}}{\theta_h}$. Note that we do not require $\kappa_{b'} \geq \kappa_b$ for proper subsets $b\subset b'\subset N $. Indeed, when $\theta_{\ell}>0$, we can have $\kappa_b> \kappa_N$.  We assume that preferences are quasilinear and players are risk-neutral expected utility maximizers.

This framework generalizes \textit{additive values} ($\kappa_b=1$ for all $b\subseteq N$) and allows for goods to be either complements or substitutes\footnote{For example, our value function can capture the case of add-on items}. This generality is not merely cosmetic and, as we will argue, has implications for the form that the seller's optimal mechanism takes. \citet{geng2005} also study multidimensional screening with such a value function, but they require that $\kappa_b$ is decreasing in the number of goods in the bundle $b$.

\medskip

\noindent \textbf{Maximal Total Surplus:} We use $\overline{\mu}:=\int_{\Theta} u(\theta,N) dF(\theta)=\int_{\overline{\Theta}} \overline{\theta} d\overline{F}(\overline{\theta})$ to denote the \textit{maximal surplus} that can be achieved from trading. Note that the weak free-disposal assumption implies that the maximal surplus is achieved by trading the grand bundle with probability one.

\medskip

\noindent \textbf{Signals:} The buyer learns about his type via a \textit{signal}. Given the linearity of our model and risk-neutrality of players, without loss, we will restrict attention to \textit{unbiased} signals $(S,G_{S\times \Theta})$. The set of signal realizations $S= \Theta$ is just the type space. $G_{S\times \Theta}\in \Delta(S\times\Theta)$ is a joint distribution over $S\times \Theta$ such that the marginal distribution of $G_{S\times \Theta}$ over $\Theta$ is $F$. We denote the marginal distribution of $G_{S\times \Theta}$ over the set of signal realizations $S$ by $ G$. 

The buyer learns about his type by observing a signal realization $s\in S$. We assume the \textit{posterior estimate} of the type is just the signal realization $s=(s_1\dots,s_n)$ itself (hence, the ``unbiased'' terminology) so $$s=\mathbb{E}_{G_{S\times \Theta}}[\theta\,|\,s]$$ for all $s$ that lie in the support of $ G$. We will refer to both the joint distribution $G_{S\times \Theta}$ and the marginal distribution $ G$ as signals since we can always convert one to the other. The buyer privately observes the signal realization.

Restricting attention to the above class of signals is without loss because the value function is linear in the type and therefore both the buyer and the seller only care about the posterior estimate of the type. To see this, suppose instead that signal realizations are drawn from an arbitrary set $S'$ and $G_{S'\times \Theta}\in \Delta(S'\times\Theta)$ is a joint distribution that has marginal distribution $F$ over $\Theta$.\footnote{We are implicitly assuming that $S'$ is a Polish space endowed with the Borel $\sigma$-algebra. This is a technical assumption which is necessary to guarantee that the conditional distributions are well defined.} If the buyer observes a signal realization $s'$, his value is determined by
$$\mathbb{E}[u(\theta,b)\,|\,s']=u(\mathbb{E}_{G_{S'\times \Theta}}[\theta\,|\,s'],b) \text{ for all } b\subseteq N.$$
So we can just transform any $(S',G_{S'\times \Theta})$ into the above unbiased form by relabeling $s'$ to $s=\mathbb{E}_{G_{S'\times \Theta}}[\theta\,|\,s']$ with $ G$ being the distribution of $\mathbb{E}_{G_{S'\times \Theta}}[\theta\,|\,s']$.



\medskip

We denote the \textit{set of signals}, by $$\mathcal{G}:=\{ G \in \Delta(S)\;|\;  G  \text{ is the marginal distribution over $S$ induced by some signal } (S,G_{S\times \Theta})\}.$$ Note that this is the set of possible distributions over posterior estimates.

\medskip

\noindent \textbf{Mechanism:} The seller chooses a mechanism. Formally, a \textit{mechanism} $\mM=(M,q,t)$ consists of a set of messages $M$, a (possibly random) allocation $q:M\to \Delta(2^N)$ and a transfer $t:M \to \mathbb{R}$. The allocation determines the likelihood of receiving the various bundles and of not being allocated any good; we sometimes use $q(m,b)$ to denote the probability that the buyer is allocated bundle $b$ when he reports message $m\in M$. 

If the buyer with posterior estimate $s$ reports $m\in M$, his expected utility is
$$\mathbb{E}_{q(m)}[u(s,b)]-t(m),$$
where the expectation is taken with respect to the random allocation.  
Because of this structure of payoffs, it is without loss to restrict the seller to deterministic transfers.

We further assume that every mechanism is such that $$\max_{m\in M}\left[\,\mathbb{E}_{q(m)}[u(s,b)]-t(m)\,\right]\geq 0,$$ for all $s\in S$. This is an individual rationality requirement. It ensures that the buyer is not forced to participate in a mechanism that gives him negative utility. Implicitly, we are also assuming that a solution to the above maximization problem exists; that is, every mechanism $\mM$ has the property that, for every $s$, there is a message that maximizes the buyer's utility.

\medskip

The buyer chooses a reporting strategy $\sigma:S\to \Delta(M)$ that maximizes his utility. We say that his strategy $\sigma$ is a \textit{best response} if his expected utility satisfies
\begin{equation}\label{eq:IC}
	U(s,\mM):=\mathbb{E}_{\sigma(s)} \left[\mathbb{E}_{q(m)}[u(s,b)]-t(m)\right]\geq \mathbb{E}_{\sigma'(s)} \left[\mathbb{E}_{q(m)}[u(s,b)]-t(m)\right]
\end{equation}
for all $s\in S$ and all other strategies $\sigma'$. For a mechanism $\mM$, $U(s,\mM)$ denotes the buyer's \textit{utility from best responding} and the \textit{set of best responses} is denoted by $\Sigma(\mM)$.

Given a signal $ G$, a mechanism $\mM$ and a best response $\sigma$, the seller's profit is given by
$$\Pi( G,\mM,\sigma):=\mathbb{E}_{ G} \left[\mathbb{E}_{\sigma(s)} \left[t(m)\right]\right].$$
The outer expectation is taken with respect to distribution over signal realizations and the inner with respect to the buyer's strategy.

\medskip

\noindent \textbf{Pure Bundling and Separate Sales:} We will refer to two special classes of mechanisms repeatedly. The first is a \textit{pure bundling mechanism} at price $\overline{p}$ which we denote by $\mM^{PB}_{\overline{p}}=(M^{PB},q^{PB}_{\overline{p}},t^{PB}_{\overline{p}})$. This is the mechanism in which, whenever the buyer purchases, he is only allowed to purchase the grand bundle $N$ at a price $\overline{p}$. Formally, we can implement this mechanism with a message space $M^{PB}=[n\theta_{\ell},n\theta_h]$ and an allocation and transfer given by

\begin{equation}\label{eq:pure_bundling}
	q^{PB}_{\overline{p}}(m,b)=\left\{\begin{array}{cl}1 & \text{ if } m\geq \overline{p} \text{ and } b=N, \\ 0 & \text{ otherwise,} \end{array}\right.\;\; \text{ and }\;\; t^{PB}_{\overline{p}}(s)=\left\{\begin{array}{cl}\overline{p}  & \text{ if } m\geq \overline{p}, \\ 0 & \text{ otherwise.} \end{array}\right. \tag{PB}
\end{equation}
In words, the buyer reports his value $m$ for the grand bundle and is allocated the grand bundle at a price of $\overline{p}$ if the report is higher than the price. Clearly, it is a best response for the buyer to truthfully report his value.

\medskip

The second is a \textit{separate sales mechanism} at prices $p=(p_1,\dots,p_n)$ which we denote by $\mM^{Sep}_p=(M^{Sep},q^{Sep}_p,t^{Sep}_p)$. Here, the seller offers a price $p_i$ for each individual good and the buyer can choose whichever bundle he likes and just pay the total price. Formally, we can implement such a mechanism with a message space $M^{Sep}=2^N$ and an allocation, transfer given by
\begin{equation}\label{eq:separate_sales}
	q^{Sep}_{p}(m,b)=\left\{\begin{array}{cl}1 & \text{ if }  m=b, \\ 0 & \text{ otherwise,} \end{array}\right. \;\;\text{ and } \;\; t^{Sep}_{p}(m)=\left\{\begin{array}{cl}\sum_{i\in m} p_i  & \text{ if } m\neq \emptyset, \\ 0 & \text{ if } m= \emptyset. \end{array}\right. \tag{Sep}
\end{equation}

\medskip
%
%

\medskip

Our aim is to derive the qualitative properties of the seller's optimal mechanism under two different information environments and to relate the solutions of each. We informally describe both here in words, the formal descriptions are in \cref{sec:buyer_opt_out} and \cref{sec:robust} respectively.

\medskip

We first derive the buyer-optimal outcome. This is a signal for the buyer and an optimal mechanism for the seller (in response to this signal) that maximize consumer surplus. The timing here is that an information designer first publicly chooses the signal. That is, the signal structure is observed by the buyer and the seller, the signal realization is private to the buyer. The seller then chooses an optimal mechanism in response.

\medskip

We then derive the seller's optimal informationally robust mechanism. This is a mechanism that maximizes the seller's profit against the worst possible signal realization. In this case, the timing is the exact opposite. The seller first chooses her mechanism. In response, nature chooses a signal and a best response for the buyer that minimizes the seller's profit.

\subsection{Discussion of the model}

We make several modeling assumptions that are worth discussing before we proceed to the analysis. The fact that we allow for an arbitrary number of goods and for goods to be either complements or substitutes makes our environment more general than the bulk of the multidimensional screening literature (especially the subset that aims to derive optimal mechanisms), which either assumes two goods, additive values or both. We chose the particular class of value functions given by \eqref{eq: value function}, because it has the feature that only the posterior estimate is relevant for the mechanism design problem. For more general value functions, the posterior estimate is no longer sufficient to determine the buyer's utility and, for every signal realization $s$, the entire posterior distribution over $\Theta$ induced by the signal $G$ would become relevant. This significantly complicates the information design part of our problem because we would need to optimize over a significantly larger set of signals. It is for precisely this reason that the bulk of the information design literature also restricts attention to linear environments.

We assume that the distribution $F$ of the buyer's type is exchangeable. This assumption is primarily for tractability; we briefly discuss the complexities introduced by more general type distributions in our concluding remarks in \cref{sec:conclusion}. Nonetheless, the fact that this assumption does not rule out positive or negative correlations makes our environment significantly more general than the bulk of the existing literature. To the best of our knowledge, apart from the aforementioned \citet{HaHaRES2020}, there are very few results that characterize the optimal mechanism when values for goods are correlated even with additional functional form assumptions about the marginal distribution of each dimension. Recent work (see, for instance, \citealp{chen2013,chen2017}) has instead sought to understand when pure bundling can dominate separate sales (even though the former may not be the optimal mechanism).

\section{The buyer-optimal outcome}\label{sec:buyer_opt_out}

In this section, we formally define and characterize a buyer-optimal outcome. We then apply this characterization to derive a comparative statics result that relates the consumer surplus to the number of goods offered for sale. We begin with a few definitions.

\medskip

\noindent \textbf{Optimal Mechanism:} For a signal $ G\in \mathcal{G}$, an \textit{optimal mechanism} maximizes the seller's profit. Formally, a mechanism $\mM$ is optimal if it has a buyer best response $\sigma\in \Sigma(\mM)$ such that
$$\Pi( G,\mM,\sigma)\geq \Pi( G,\mM',\sigma')$$
for any other mechanism $\mM'$ and buyer best response $\sigma'\in\Sigma(\mM').$ 
Here we make the implicit assumption that the buyer chooses the best response that maximizes the seller's revenue (this could have a material impact on profits since, in particular, $ G$ may have atoms). This is the standard assumption in mechanism design. In other words, for a given signal $ G$, the optimal mechanism is the solution to a standard multidimensional screening problem. Despite being standard, we make this tie-breaking assumption explicit to provide a clear contrast with the optimal informationally robust mechanism that we define in \cref{sec:robust} (in which the buyer breaks ties against the seller).

\medskip

\noindent \textbf{Outcome:} We will refer to a pair $\left( G,\ \mM\right)$ as an \emph{outcome}, whenever $ G\in \mathcal{G}$ is a signal and $\mM$ is an \textit{optimal mechanism} for the seller in response to distribution $ G$. 

\medskip

\noindent \textbf{Buyer-Optimal Outcome:} Our goal in this section is to characterize the \emph{buyer-optimal outcome} $\left( G^*,\ \mM^*\right)$ that maximizes the buyer's surplus across all outcomes. Formally, we solve
\begin{align*}
	&\max_{\mM, G\in \mathcal{G}} \mathbb{E}_{ G}\left[U(s,\mM)\right] \\
	& \text{ such that }\\
	& \mM \text{ is an optimal mechanism for the signal } G.
\end{align*}
The constraint in this maximization problem is well defined because a seller optimal mechanism always exists for every signal (see, for instance, \citet{balder1996}). Our proof explicitly constructs the signal $ G^*$ in the buyer-optimal outcome, so we show that the maximum for the objective function is obtained.

\medskip

Before stating the result, we highlight a specific subclass of signals. These will play an important role in the derivations of both the buyer-optimal outcome and the optimal informationally robust mechanism.

\medskip

\noindent \textbf{Perfectly correlated signal:} We say that a signal $ G\in\mathcal{G}$ is maximally positively correlated across dimensions, or simply \textit{perfectly correlated}, if it is distributed along the diagonal $\{(s_1,\dots,s_n)\in S\;|\; s_1=\cdots=s_n\}$.

\medskip

We are now in a position to present the first of our two main results.

\begin{theorem}\label{thm:buyer_opt}
	There exists a buyer-optimal outcome $\left( G^*,\ \mM^*\right)$ which has the following properties.
	\begin{enumerate}[leftmargin=*]
		\item \emph{Seller Best-Response:} $\mM^*$ is a pure bundling mechanism so $\mM^*=\mM^{PB}_{\overline{p}^*}$ for some $\overline{p}^*$.
		\item \emph{Signal}: $ G^*$ is perfectly correlated.
		\item \emph{Total Surplus:} The buyer is allocated the grand bundle with probability one so trade is efficient.
	\end{enumerate}
\end{theorem}

Before proceeding to the proof, it is worth briefly discussing this result. We begin by providing some intuition for why pure bundling emerges as the optimal mechanism. It turns out that perfect information about his type is not optimal for the buyer because this allows the seller to screen effectively. To prevent this, the signal in the buyer-optimal outcome injects two different kinds of noise into learning. First, it introduces perfect correlation. This effectively makes the type space one-dimensional and reduces the seller's ability to screen across dimensions. It is generically possible to construct informative perfectly correlated signals even though the $\theta_i$'s might be independently distributed or even negatively correlated.\footnote{Of course, a completely uninformative signal always induces perfect correlation.} 
In our setting, suppose that the buyer perfectly learned his grand bundle value $\theta_1+\cdots+\theta_n$ but nothing further. Then, because the distribution $F$ is exchangeable, the buyer's posterior estimate for each dimension will be identical and will simply be $(\theta_1+\cdots+\theta_n)/n$. 

Reducing the seller's ability to screen by introducing correlation could still harm the buyer as it might simultaneously lead to a reduction of total surplus. This can be prevented by injecting further noise into the signal: instead of telling the buyer his exact grand bundle value, the signal provides a noisy estimate. Loosely speaking, perfect correlation effectively reduces the seller's problem to its one-dimensional counterpart, and hence we can adapt the methods from \citet{RoSzAER2017} (who study the sale of a single good) to show that it is possible to construct such a signal so that trade of the grand bundle always happens. The latter is the efficient outcome because of the weak free-disposal assumption. Finally, if trade always happens, the seller must be best-responding by offering a pure bundling mechanism where the price for the grand bundle is the minimum of the support of the distribution of the posterior estimate of the grand bundle value induced by $G^*$. 

It might seem surprising that the signal in the buyer-optimal outcome inserts so much noise into learning; after all, as a result, the buyer might frequently purchase the grand bundle even when his true value is lower than the price. Both types of noise are necessary because, compared to the one-dimensional case, the larger set of mechanisms at the seller's disposal implies that a greater amount of noise is necessary to prevent effective screening. The buyer-optimal outcome exactly balances these two countervailing forces: making a good purchasing decision and preventing effective screening. 

\vspace*{.1in}

In fact, the buyer-optimal outcome is the worst outcome for the seller.

\begin{corollary}\label{cor:seller_min}
	The seller's profit in any outcome is weakly greater than her profit $\pi^*$ in any buyer-optimal outcome. As a consequence, trade happens with probability one in every buyer-optimal outcome.
\end{corollary}
In words, this corollary (proved in \cref{app:proofs}) says that every buyer-optimal outcome not only maximizes consumer surplus, it also minimizes producer surplus. The first part of the above corollary implies that the seller's profit $\pi^*$ must be the same in every buyer-optimal outcome. Statement (3) in \cref{thm:buyer_opt} then implies that trade must be efficient in every buyer-optimal outcome.

This corollary also implies that the seller's revenue in the buyer-optimal outcome is the solution to the min-max problem where an adversarial nature first picks the signal and the seller then chooses an optimal mechanism in response. In other words,
\begin{equation}\label{eq:minmax}
	\pi^*=\min_{ G\in\mathcal{G}}\max_{\mM,\sigma\in\Sigma(\mM)}\Pi( G,\mM,\sigma).
\end{equation}
This fact will play an important role in the derivation of the optimal informationally robust mechanism.

\medskip

We make one last observation before presenting the proof of \autoref{thm:buyer_opt}. Note that the buyer-optimal outcome is not unique and, in particular, pure bundling need not be the unique seller best response. Specifically, if $\kappa_b \leq \kappa_N=1$ for all $b\subset N$, then it will also be optimal for the seller to offer a separate sales mechanism where the price vector for the goods is simply the minimum of the support of $ G^*$. This is because the buyer will still prefer to always buy the grand bundle when faced with this mechanism. However, this is no longer the case when $\kappa_b>1$. We illustrate this in the following example.

\medskip

\noindent \textbf{Example}. Suppose there are two goods and that $\kappa_{\{i\}}>1$ for both $i\in \{1,2\}$. We now argue that separate sales cannot be the seller's optimal mechanism in \textit{any} buyer-optimal outcome. As a contradiction, suppose that there is a buyer-optimal outcome $\left(G^{\star},\mM^{Sep}_p\right)$ where $\mM^{Sep}_p$ is a separate sales mechanism at prices $p=(p_1,p_2)$. \cref{cor:seller_min} implies that trade must be efficient and so the buyer must always purchase the grand bundle. This in turn implies that, for every $\varepsilon>0$, we must have $G^{\star}(\{(s_1,s_2)\;|\; (s_1+s_2)-(p_1+p_2)\leq \varepsilon\})>0$. In words, there must be a positive mass of buyer types whose grand bundle value is just above the total price for the grand bundle as, otherwise, the seller could earn a greater profit by instead offering a pure bundling mechanism at the higher price $p_1+p_2+\varepsilon$.

Now consider the positive mass of types that satisfy $0\leq (s_1+s_2)-(p_1+p_2)\leq \varepsilon$. The buyer prefers to purchase the grand bundle instead of just good $i$ whenever $\kappa_{\{i\}} s_i -p_i \leq (s_1+s_2)-(p_1+p_2)$. Adding up over the two goods, we get
$$(\kappa_b-1) (p_1+p_2)\leq (\kappa_b-1) (s_1+s_2)\leq (\kappa_{\{1\}}-1) s_1+ (\kappa_{\{2\}}-1) s_2\leq (s_1+s_2)-(p_1+p_2)\leq\varepsilon$$
where $\kappa_b=\min\{\kappa_{\{1\}},\kappa_{\{2\}}\}$. For small enough $\varepsilon>0$, the above inequality cannot be true since we must have $p_1+p_2>0$. This provides the requisite contradiction because it implies that there will be a positive mass of buyer types that will strictly prefer to not buy the grand bundle thereby implying trade is not efficient in the buyer-optimal outcome $\left(G^{\star},\ \mM^{Sep}_p\right)$.

\subsection{Proof of \cref{thm:buyer_opt}}\label{sec:proof_main_result}

In this subsection, we prove \cref{thm:buyer_opt}. Readers who are not interested in the details of the proof can skip to \cref{sec:comp_statics} without loss of continuity. However, the proof of \cref{thm:robust} (the characterization of the optimal informationally robust mechanism) builds on the arguments that follow.

\medskip

We begin by introducing some additional notation and terminology. Given $s\in S$, we use $\bar{s}=s_1+\cdots+s_n$ to denote the sum. The set of all such $\overline{s}$ is denoted by $\overline{S}$; note that $\overline{S}=\overline{\Theta}$ (because $S=\Theta$) but we use distinct notation nonetheless to distinguish the sum of the signal realization vector from the sum of the type vector. Every signal $G\in\mathcal{G}$ induces a distribution $\overline{G}\in \Delta(\overline{S})$ over the posterior estimates $\bar{s}$ of the grand bundle value. We use $\overline{\mathcal{G}}$ to denote the set of these distributions of grand bundle estimates that are induced by some signal $G\in\mathcal{G}$.

Given two distributions $\overline{G},\overline{G}'\in \Delta(\overline{S})$, we say that $\overline{G}$ is a \textit{mean-preserving spread} of $\overline{G}'$ or $\overline{G} \succsim \overline{G}'$ if 
\begin{equation}\label{eq:sosd}  \int_{n\theta_{\ell}}^{\overline{s}} \overline{G}(x) \td x \geq \int_{n\theta_{\ell}}^{\overline{s}} \overline{G}' (x) \ \td x\quad \text{ for all } \overline{s}\in[n\theta_{\ell},n\theta_h]
	\text{ with equality for } \overline{s}=n\theta_h.
\end{equation}

\medskip

We now establish some properties of the set $\overline{\mathcal{G}}$.

\begin{lemma}\label{lemma:corr_signal} 
	The set of distributions over grand bundle estimates induced by signals in $\mathcal{G}$ has the following properties:
	\begin{enumerate}
		\item[(i)] $\overline{G}$ is a distribution over grand bundle estimates iff $\overline{F}$ is a mean-preserving spread of $\overline{G}$, or equivalently, $$\overline{\mathcal{G}}=\left\{\overline{G} \in \Delta(\overline{S}) \; | \; \overline{F}\succsim \overline{G} \right\}.$$
		\item[(ii)] For every signal $G\in \mathcal{G}$, there exists a signal $G'\in \mathcal{G}$ that is perfectly correlated such that $G$ and $G'$ induce the same distribution $\overline{G}=\overline{G}'$ over posterior grand bundle estimates.
	\end{enumerate}
\end{lemma}
\begin{proof}[Proof of \cref{lemma:corr_signal}]
	We begin by defining one-dimensional signals that only provide the buyer information about his grand bundle value. These are (unbiased) signals $(\overline{S},H_{\overline{S}\times \overline{\Theta}} )$ where the set of signal realizations $\overline{S}$ is just the set of possible grand bundle values and $H_{\overline{S}\times \overline{\Theta}}\in \Delta(\overline{S}\times\overline{\Theta})$ is a joint distribution over $\overline{S}\times \overline{\Theta}$ such that the marginal distribution of $H_{\overline{S}\times \overline{\Theta}}$ over $\overline{\Theta}$ is $\overline{F}$ and  
	$$\overline{s}=\mathbb{E}_{H_{\overline{S}\times \overline{\Theta}}} [\overline{\theta}\,|\,\overline{s}]$$
	for all $\overline{s}$ in the support. We refer to such one-dimensional signals as \emph{grand-bundle signals}.
	
	\medskip
	
	We use $H$ to denote the marginal distribution of $H_{\overline{S}\times \overline{\Theta}}$ over the set of signal realizations $\overline{S}$ and use $\mathcal{H}$ to denote the set of all such distributions $H$ over grand bundle estimates. As we argued when we defined signals $(S,G_{S\times\Theta})$ for the type vector, it is without loss to restrict attention to such unbiased signals.
	
	Now note that elements of $\mathcal{H}$ are just cdfs of real-valued random variables. We defined this set because it has a well known characterization
	\begin{equation}\label{eq:signal_def}
		\mathcal{H}=\left\{H \in \Delta(\overline{S}) \; | \; \overline{F}\succsim H \right\},
	\end{equation}
	or, in words,   $H$ is the marginal distribution of posterior estimates of grand bundle values for some grand bundle signal   iff $\overline{F}$ is a mean-preserving spread of $H$. We will use this characterization to establish the first part of the lemma.
	
	\medskip
	
	We first argue that $\overline{\mathcal{G}}\subseteq \mathcal{H}$. To see this, observe that signal $G_{S\times \Theta}$ induces a joint distribution $\overline{G}_{\overline{S}\times \overline{\Theta}}$ over $\overline{S}\times \overline{\Theta}$ such that the marginal distribution of $\overline{G}_{\overline{S}\times \overline{\Theta}}$ over $\overline{\Theta}$ is $\overline{F}$ and the marginal distribution over $\overline{S}$ is $\overline{G}$. Formally, this is the image measure of $G_{S\times \Theta}$ generated by the mapping $a(s,\theta)=(s_1+\cdots+s_n,\theta_1+\cdots+\theta_n)$ which implies that, for any measurable set $A\subset \overline{S}\times\overline{\Theta}$, we have $\overline{G}_{\overline{S}\times \overline{\Theta}}(A)=G_{S\times \Theta}(a^{-1}(A))$. Moreover, observe that
	$$\mathbb{E}_{\overline{G}_{\overline{S}\times \overline{\Theta}}}\left[\overline{\theta}|\overline{s}\right]=\mathbb{E}_{G_{S\times \Theta}}\left[\overline{\theta}|\overline{s}\right]=\mathbb{E}_{G_{S\times \Theta}}\left[\mathbb{E}_{G_{S\times \Theta}}[\theta_1+\cdots+\theta_n|s]|\overline{s}\right]=\mathbb{E}_{G_{S\times \Theta}}[s_1+\cdots+s_n|\overline{s}]=\overline{s}$$
	for all $\overline{s}$ in the support. Therefore the marginal distribution $\overline{G}$ over $\overline{S}$ induced by $\overline{G}_{\overline{S}\times \overline{\Theta}}$ satisfies $\overline{G}\in \mathcal{H}$.
	
	\medskip
	
	To complete the proof of part (i), we need to show that $\mathcal{H}\subseteq \overline{\mathcal{G}}$. We will in fact also show part (ii) by arguing that, for every $H\in \mathcal{H}$, there exists a perfectly correlated signal $G\in \mathcal{G}$ such that the distribution it induces on grand bundle estimates satisfies $\overline{G}=H$.
		
	By definition, $H$ is the marginal distribution over $\overline{S}$ corresponding to a joint distribution $H_{\overline{S}\times\overline{\Theta}}\in \Delta(\overline{S}\times\overline{\Theta})$ that has marginal distribution $\overline{F}$ over $\overline{\Theta}$ and that satisfies $\overline{s}=\mathbb{E}_{H_{\overline{S}\times \overline{\Theta}}} [\overline{\theta}\,|\,\overline{s}]$. We use $H_{\overline{S}\times\overline{\Theta}}$ to  define a family of conditional distributions $\hat{H}(\cdot|\theta)\in \Delta(\overline{S})$ as follows
	\begin{equation}\label{eq:lem1}
		\hat{H}(\cdot|\theta):=H(\cdot|\theta_1+\cdots+\theta_n).
	\end{equation}	
	This combined with the distribution $F$ over $\Theta$ generates a joint distribution $\hat{H}_{\overline{S}\times \Theta} $ over $\overline{S}\times \Theta$ whose marginal distributions over $\overline{S}$ and $\Theta$ are $\hat{H}=H$ and $F$ respectively.
	
	Now observe that
	$$\mathbb{E}_{\hat{H}_{\overline{S}\times \Theta} }\left[\theta_1+\cdots+\theta_n| \overline{s}\right] =\overline{s}$$
	for all $\overline{s}$ in the support. This is a consequence of the definition \eqref{eq:lem1} of $\hat{H}_{\overline{S}\times \Theta} $ and from the fact that $\overline{s}=\mathbb{E}_{H_{\overline{S}\times \overline{\Theta}}} [\overline{\theta}\,|\,\overline{s}]$.
	
	Given the joint distribution $\hat{H}_{\overline{S}\times \Theta} $, we can derive the conditional distribution $\hat{H}(\cdot | \overline{s})$  over $\Theta$. Now observe that the conditional distribution $\hat{H}(\cdot | \overline{s})$ is exchangeable. This follows from the definition of $\hat{H}_{\overline{S}\times \Theta} $ and because $F$ is assumed to be exchangeable. This in turn implies
	$$\mathbb{E}_{\hat{H}}\left[\theta_i| \overline{s}\right] =\frac{\overline{s}}{n} \quad \text{ for all } i\in\{1,\dots,n\}.$$
	
	Now define a joint distribution $G_{S\times \Theta}$ over $S\times\Theta$ that is the image measure of $\hat{H}_{\overline{S}\times \Theta} $ generated by the mapping $\hat{a}(\overline{s},\theta)=\left(\frac{\overline{s}}{n},\dots, \frac{\overline{s}}{n},\theta \right)$. Formally, for any measurable $\hat{A}\subseteq S\times\Theta$, we have $G_{S\times \Theta}(\hat{A})=\hat{H}_{\overline{S}\times \Theta} (\hat{a}^{-1}(\hat{A}))$. By construction, the marginal distribution of $G_{S\times \Theta}$ over $\Theta$ is $F$, the marginal distribution $ G$ over $S$ is distributed along the diagonal $\{(s_1,\dots,s_n)\in S\;|\; s_1=\cdots=s_n\}$ and so is perfectly correlated. 
	
	Observe that
	$$\mathbb{E}_{G_{S\times \Theta}}\left[\theta_i\,\bigg|\, s=\left(\frac{\overline{s}}{n},\dots, \frac{\overline{s}}{n}\right)\right] =\mathbb{E}_{\hat{H}_{\overline{S}\times \Theta} }\left[\theta_i| \overline{s}\right] =\frac{\overline{s}}{n} \quad \text{ for all } i\in\{1,\dots,n\}$$
	and so $ G\in \mathcal{G}$ or, in words, that $ G$ is an unbiased signal. Finally, by construction, the distribution $\overline{G}$ over posterior grand bundle estimates induced by $G$ satisfies $\overline{G}=\hat{H}=H$ which completes the proof.
\end{proof}

\medskip

We now define the special class of truncated Pareto distributions that we will employ in the proof of \cref{thm:buyer_opt} below. These are defined as
\begin{equation} \label{eq: TPD} 
	H_{\alpha,\beta}(\overline{\theta})=\left\{\begin{array}{cl}
		0 \ & \ \text{ if } \ \overline{\theta}<\ua, \\
		1-\alpha/\overline{\theta} \ & \ \text{ if } \ \overline{\theta}\in [\alpha,\ \beta),\\
		1 \ & \ \text{ if } \overline{\theta} \geq \beta,\\
	\end{array}\right.
\end{equation}
where $\alpha\leq \beta$. Truncated Pareto distributions are supported on $[\alpha,\beta]$, are continuous on $(\alpha,\beta)$ and have an atom of size $\frac{\alpha}{\beta}$ at the truncation point $\beta$. When $\alpha=\beta$, $H_{\alpha,\alpha}$ is the degenerate distribution with an atom of size $1$ at $\alpha$.

\medskip

We use $\mathcal{\overline{G}}^{P} \subset \mathcal{\overline{G}}$ to denote the subset of distributions over grand bundle estimates (induced by signals in $\mathcal{G}$) that are truncated Pareto distributions. Formally,
$$\mathcal{\overline{G}}^{P} =\left\{ \overline{G} \in \mathcal{\overline{G}}  \; | \; \overline{G} = H_{\alpha,\beta}\;\; \text{ for some } \;\;n\theta_{\ell}\leq \alpha\leq \beta \leq n\theta_h\; \right\}.$$
This set is non-empty because it includes the distribution $H_{\overline{\mu},\overline{\mu}}\in \mathcal{\overline{G}}$ over grand bundle estimates induced by the completely uninformative signal. 

It is easy to show that there is a continuous, strictly decreasing function $\beta(\alpha)$ such that every element of $\mathcal{\overline{G}}^{P}$ is of the form $H_{\alpha,\beta(\alpha)}$. Hence, in what follows, we simplify notation and only use the lower bound of the support $H_{\alpha}\in \mathcal{\overline{G}}^{P}$ to denote a truncated Pareto distribution over grand bundle estimates induced by a signal, since this alone suffices to describe this distribution. In other words, whenever we refer to a distribution $H_{\alpha}$ we are implicitly assuming $H_{\alpha}\in \mathcal{\overline{G}}^{P}$ and that $H_{\alpha}=H_{\alpha,\beta(\alpha)}$.

\medskip

\citet{RoSzAER2017} showed that the class of truncated Pareto distributions can be used to characterize the buyer-optimal outcome for a single good. The truncated Pareto distribution $H_{\alpha,\beta}$ has the property that all pure bundling mechanisms with price $\overline{p}$ in $[\alpha,\beta]$ yield the same profit $\alpha$. Since their work, the properties of this class of distributions have been exploited in several information design papers.

\medskip

We are now ready to prove \cref{thm:buyer_opt}.

\medskip

\begin{proof}[Proof of \cref{thm:buyer_opt}]

We prove the theorem in two steps. The first step adapts the proof of Lemma 1 in \citet{RoSzAER2017} to our multidimensional context. While the proofs are similar, we reproduce the complete argument here so that our proof is self-contained, and we flag the key differences.

\medskip

\noindent \underline{Step 1:} Consider an arbitrary outcome $\left( G,\ \mM\right)$ at which the seller's profit is $\pi\in [n\theta_{\ell},n\theta_h]$. There exists a perfectly correlated signal $ G' \in \mathcal{G}$ that satisfies the following properties.
\begin{enumerate}[label=\roman*]
	\item $ G' \in \mathcal{G}$ induces a truncated Pareto distribution $\overline{G}' =H_{\pi}\in \mathcal{\overline{G}}^P$ over grand bundle estimates.
	\item The buyer's consumer surplus under signal $ G'$ is weakly higher (than that under outcome $\left( G,\ \mM\right)$) when the seller offers the pure bundling mechanism $\mM^{PB}_{\pi}$ at price $\pi$.
	\item Pure bundling mechanism $\mM^{PB}_{\pi}$ yields the seller the highest profit in the set of pure bundling mechanisms.
\end{enumerate}

\medskip

\noindent \underline{Proof of Step 1:} Let $\overline{G}$ be the distribution on the grand bundle estimates induced by $ G$. Given this signal, if the seller offers a pure bundling mechanism $\mM^{PB}_{\overline{p}}$ at price $\overline{p}$, the buyer purchases the grand bundle if his posterior estimate of the grand bundle value satisfies $\overline{s}\geq \overline{p}$ (where note that we assume the buyer purchases when $\overline{s}= \overline{p}$). Therefore, the seller's profit from this pure bundling mechanism satisfies
$$\overline{p}\left[1-\overline{G}(\overline{p})+\delta_{\overline{G}}(\overline{p})\right]\leq \pi \quad \Longleftrightarrow \quad 1-\frac{\pi}{\overline{p}}\leq \overline{G}(\overline{p})-\delta_{\overline{G}}(\overline{p})$$
where $\delta_{\overline{G}}(\overline{p})$ is the mass of distribution $\overline{G}$ at $\overline{p}$ (if $\overline{G}$ has no atom at $\overline{p}$, this will be 0). The left inequality follows from the fact that $\mM$ is an optimal mechanism in response to signal $ G$.

Observe that since $\delta_{\overline{G}}(\overline{p})\geq 0$, we can conclude from the right inequality that $H_{\pi,n\theta_h}(\overline{p})\leq \overline{G}(\overline{p})$ for all $\overline{p}\in [n\theta_{\ell}, n\theta_h]$ because, recall that,
$$H_{\pi,n\theta_h}(\overline{s})=\left\{\begin{array}{cl}
	0 \ & \ \text{ if } \ \overline{s}<\pi, \\
	1-\frac{\pi}{\overline{s}} \ & \ \text{ if } \ \overline{s}\in [\pi,n\theta_h),\\
	1 \ & \ \text{ if } \overline{s} \geq n\theta_h.\\
\end{array}\right.$$
In words, $H_{\pi,n\theta_h}$ first-order stochastically dominates $\overline{G}$. This implies
\begin{equation}\label{eq:RS_ineq} 
	\int_{n\theta_{\ell}}^{n\theta_h} \overline{s}\, \td H_{\pi,n\theta_h}(\overline{s}) \geq \int_{n\theta_{\ell}}^{n\theta_h} \overline{s} \, \td \overline{G}(\overline{s}) = \overline{\mu} \geq \pi= \int_{n\theta_{\ell}}^{n\theta_h} \overline{s}\, \td H_{\pi,\pi}(\overline{s}).
\end{equation}
The left inequality follows from the first-order stochastic dominance we established above. The left equality is a consequence of the fact that $ G$ is a signal and so $\int_{n\theta_{\ell}}^{n\theta_h} \overline{s} \, \td \overline{G}(\overline{s})=\int_{n\theta_{\ell}}^{n\theta_h} \overline{\theta} \, \td \overline{F}(\overline{\theta}) = \overline{\mu}$. The right equality follows from the fact that the distribution $H_{\pi,\pi}$ is just the Dirac measure that has an atom of mass 1 at $\pi$.

The simple but key observation in the proof of this step (relative to the one dimensional case analyzed in \citet{RoSzAER2017}) is the right inequality. This follows from the weak free-disposal assumption on the buyer's value function, which implies that always trading the grand bundle generates the maximal surplus, and therefore we must have $\pi \leq \overline{\mu}$. This implies that there exists a pure bundling mechanism that would generate a higher consumer surplus.

Since $\int_{n\theta_{\ell}}^{n\theta_h} \overline{s}\, \td H_{\pi,z}(\overline{s})$ is continuous and strictly increasing in $z\in[\pi,n\theta_h]$, the inequality \eqref{eq:RS_ineq} combined with the intermediate value theorem imply that there is a unique $\tau\in[\pi,n\theta_h]$ such that $\int_{n\theta_{\ell}}^{n\theta_h} \overline{s}\, \td H_{\pi,\tau}(\overline{s})  = \overline{\mu}$. In words, the truncated Pareto distribution $H_{\pi,\tau}$ has the same mean as $\overline{F}$ and $\overline{G}$.

\medskip

We will now argue that  $\overline{G}$ is a mean-preserving spread of $H_{\pi,\tau}$ or, in our notation, that $\overline{G}\succsim H_{\pi,\tau}$. To see this, first observe that $H_{\pi,\tau}(\overline{s})=H_{\pi,n\theta_h}(\overline{s})\leq \overline{G}(\overline{s})$ on $\overline{s}\in[n\theta_{\ell},\tau)$. This implies that for any $z\in [n\theta_{\ell}, \tau)$  
$$\int_{n\theta_{\ell}}^z \overline{G}(\overline{s}) \td \overline{s} \geq \int_{n\theta_{\ell}}^z H_{\pi,\tau}(\overline{s}) \td \overline{s}, $$
and for any $z\in [\tau,n\theta_h]$,
\begin{equation} \int_{n\theta_{\ell}}^z \overline{G}(\overline{s}) \td \overline{s}=\overline{\mu}-\int_z^{n\theta_h} \overline{G}(\overline{s}) \td \overline{s} \geq \overline{\mu}-\int_z^{n\theta_h} H_{\pi,\tau}(\overline{s}) \td \overline{s}=\int_0^z H_{\pi,\tau}(\overline{s}) \td \overline{s}, 
\end{equation}
where the inequality follows from $H_{\pi,\tau}(\overline{s})=1\geq \overline{G}(\overline{s})$ for all $\overline{s}\geq \tau.$

\medskip

Part (i) of \cref{lemma:corr_signal} shows that $\overline{F} \succsim \overline{G}$ and therefore $\overline{G}\succsim H_{\pi,\tau}$ implies $\overline{F} \succsim H_{\pi,\tau}$. In other words, $H_{\pi}\in\mathcal{\overline{G}}$ is a distribution over posterior grand bundle estimates induced by a signal.

\medskip

Part (ii) of \cref{lemma:corr_signal} implies that we can find a signal $ G'\in \mathcal{G}$ that is perfectly correlated and that induces a distribution $\overline{G}'$ over posterior grand bundle estimates that satisfies $\overline{G}'=H_{\pi}\in\mathcal{\overline{G}}$.

\medskip

Finally, observe that the pure bundling mechanism $\mM^{PB}_{\pi}$ at price $\pi$ is an optimal pure bundling mechanism (though not necessarily an optimal mechanism) for the seller when the buyer's posterior grand bundle estimate is distributed according to $H_{\pi}$. Since this price is the minimum of the support of $H_{\pi}$, trade is efficient and therefore the buyer's surplus is weakly higher in this case than compared to outcome $( G, \mM)$. This is because profits are the same in both cases, and the maximal surplus $\overline{\mu}$ is generated by $\mM^{PB}_{\pi}$.

\medskip

\noindent \underline{Step 2:} There exists a buyer-optimal outcome $\left( G^*,\ \mM^*\right)$ such that:
\begin{enumerate}[label=\roman*]
	\item Trade is efficient.
	\item $G^*$ is perfectly correlated.
	\item $G^*$ induces a truncated Pareto distribution $\overline{G}^*\in \mathcal{\overline{G}}^{P}$ over grand bundle estimates.
	\item $\mM^*$ is a pure bundling mechanism.
\end{enumerate}

\medskip

\noindent \underline{Proof of Step 2:} Let $\left( G,\ \mM\right)$ be an outcome. Then by Step 1, there exists a signal $ G'\in \mathcal{G}$ that is perfectly correlated such that $\overline{G}'=H_{\alpha}\in \mathcal{\overline{G}}$ for some $n\theta_{\ell}\leq \alpha\leq n\theta_h$ and the buyer is weakly better off under this signal when the seller offers a pure bundling mechanism $\mM^{PB}_{\alpha}$ at price $\alpha$. We now argue that $\mM^{PB}_{\alpha}$ is an optimal mechanism for the seller in response to $ G'$.

\medskip

%

So suppose that the buyer learns via signal $ G'$. Then, for any mechanism $\mM'=(M',q',t')$ and any buyer best response $\sigma'\in \Sigma(\mM')$, the seller's revenue satisfies $$\int_{s\in S} \mathbb{E}_{\sigma'(s)}\,[t(m)]\, \td  G'(s)=\int_{\overline{s}\in \overline{S}} \mathbb{E}_{\sigma'\left(\frac{\overline{s}}{n},\dots,\frac{\overline{s}}{n}\right)}\,[t(m)]\, \td\overline{G}'(\overline{s})$$
because $ G'$ is perfectly correlated and so the distribution is supported on the diagonal. Observe that this implies that, given the signal $ G'$, the seller is effectively solving a one-dimensional mechanism design problem. The Revelation Principle applies so the seller's problem is equivalent to choosing a direct mechanism $\overline{\mM}=(\overline{M},\overline{q},\overline{t})$ with $\overline{M}=\overline{S}$, $\overline{q}:\overline{S}\to \Delta(2^N)$, $\overline{t}:\overline{S}\to \mathbb{R}$ (defined on a one-dimensional type space where the type is the grand bundle value) to solve
\begin{equation}\label{eq:step2_relaxed}
	\begin{array}{l}
		\max_{(\overline{q},\overline{t})} \bigintsss_{\overline{s}\in \overline{S}} t\left(\overline{s}\right) d\overline{G}'\left(\overline{s}\right)\\
		\text{subject to}\\
		\mathbb{E}_{\overline{q}\left(\overline{s}\right)}\left[u\left(\frac{\overline{s}}{n},\dots,\frac{\overline{s}}{n},b\right)\right]-\overline{t}\left(\overline{s}\right)\geq \mathbb{E}_{\overline{q}\left(\hat{s}\right)}\left[u\left(\frac{\overline{s}}{n},\dots,\frac{\overline{s}}{n},b\right)\right]-\overline{t}\left(\hat{s}\right) \quad \text{ for all }\overline{s},\hat{s}\in \overline{S} \;\;\text{ and}\\
		\mathbb{E}_{\overline{q}\left(\overline{s}\right)}\left[u\left(\frac{\overline{s}}{n},\dots,\frac{\overline{s}}{n},b\right)\right]-\overline{t}\left(\overline{s}\right)\geq 0 \quad \text{ for all }\overline{s}\in \overline{S}.
	\end{array}
\end{equation}
The constraints are simply the standard incentive compatibility and individual rationality constraints.

Now observe that, for $\overline{s}>0$, the ratio of the value of the grand bundle $N$ to any other bundle $b$ is a constant because
$$\frac{u\left(\frac{\overline{s}}{n},\dots,\frac{\overline{s}}{n},N\right)}{u\left(\frac{\overline{s}}{n},\dots,\frac{\overline{s}}{n},b\right)}=\frac{\overline{s}}{\kappa_b \frac{|b|}{n}\overline{s}}=\frac{n}{\kappa_b |b|}.$$

\medskip

We can then directly apply Proposition 1 in \citet{HaHaRES2020}\footnote{Informally, their result states that a pure bundling mechanism is optimal when types are one-dimensional and the ratio of the value of the grand bundle $N$ to every other bundle is non-increasing.} to conclude that a pure bundling mechanism $\mM^{PB}_{\overline{p}}$ at some price $\overline{p}$ solves \eqref{eq:step2_relaxed}. Because $\overline{G}'=H_{\alpha}$, we can then conclude that the pure bundling mechanism $\mM^{PB}_{\alpha}$ at the price $\alpha$ is a best response of the seller to $ G'$.

\medskip

Now define,
$$\alpha^*:=\min\left\{\alpha\;\big|\; H_{\alpha}\in\mathcal{\overline{G}} \right\}.$$
In words, $\alpha^*$ is the lowest minimum of the support of a truncated Pareto distribution over grand bundle estimates induced by a signal; the existence of this minimum can be established by a simple continuity argument. By \cref{lemma:corr_signal}, there exists a perfectly correlated signal $ G^*$ such that the distribution over grand bundle estimates that it induces satisfies $\overline{G}^*=H_{\alpha^*}$ and, by the above argument, the pure bundling mechanism $\mM^{PB}_{\alpha^*}$ at the price $\alpha^*$ is an optimal mechanism for the seller in response to $ G^*$.

\medskip

By construction, the outcome $( G^*,\mM^{PB}_{\alpha^*})$ satisfies all the properties (1)--(3) of the theorem. Moreover, from Step 1, it is buyer-optimal. If there were another outcome that generated strictly higher consumer surplus, Step 1 implies that we would be able to find a truncated Pareto distribution $H_{\alpha}\in\mathcal{\overline{G}}$ where $\alpha<\alpha^*$ which is a contradiction. This completes the proof of this step and of the theorem.
\end{proof}

\subsection{Comparative statics}\label{sec:comp_statics}

In this section, we apply \cref{thm:buyer_opt} to derive a comparative static relating consumer surplus to the number of goods. We begin with some context. It is clearly beneficial for the monopolist to have the ability to screen over all $n$ goods as opposed to having to set a price for each good individually. This is because maximizing profits over a strictly larger set of mechanisms must achieve a weakly higher profit. However, as \citet{salinger1995} observes, increased profits need not be at the expense of consumer surplus. For instance, consider a buyer with additive values for $n=2$ goods where his value for each good is independently and uniformly distributed on $[0,1]$. Here, the optimal separate sales mechanism is to charge a price of $\frac12$ for each good. Now, suppose that in addition to selling the goods individually, the seller was allowed to pure bundle. She would choose to do the latter, and the optimal pure bundling mechanism is to charge a price of $\sqrt{\frac23}<\frac12+\frac12$ for the grand bundle. The latter mechanism (which exploits the fact that there are multiple goods) leads to \textit{both} higher profits and consumer surplus.

By contrast, \citet{bakos1999} derive a limit result that shows the seller can extract all the surplus from a buyer with additive, iid values when the number of goods $n\to \infty$. They use a law of large numbers to argue that the value of the grand bundle divided by the number of goods $n$ converges, and so the seller can extract all the surplus by just offering a pure bundling mechanism at a price of $n$-times that limit. As we argued in the introduction, because it is hard to characterize the optimal mechanism, we are not aware of any general results for finitely many goods that describe whether the seller's ability to screen across multiple dimensions hurts consumers. The goal of this section is to show that such an analysis is possible for the buyer-optimal outcome.

\medskip

In this section, we restrict attention to the case of additive values and we assume each $\theta_i$ is iid with distribution $\tilde{F}$ that has a positive density (so the joint distribution is $F=\tilde{F}\times\cdots\times \tilde{F}$). We denote $\tilde{\mu}=\mathbb{E}_{\tilde{F}}[\theta_i]$ to be the mean of each dimension of the agent's type. Because we will vary the number of goods we define $\tilde{S}=[\theta_{\ell},\theta_h]$ and use $\tilde{S}^n:=[\theta_{\ell},\theta_h]^n$ (instead of just $S$) to denote the set of possible signal realizations. Additionally, we will use $\overline{S}_n:=[n\theta_{\ell},n\theta_h]$ to denote the set of possible grand bundle posterior estimates.

\cref{thm:buyer_opt} states that there is a buyer-optimal outcome $\left(G^*_{n},\, \mM^{PB}_{\overline{p}^*_n}\right)$ where $G^*_{n}$ is perfectly correlated and $\mM^{PB}_{\overline{p}^*_n}$ is a pure bundling mechanism at price $\overline{p}^*_n$. As we explained in the discussion following \cref{thm:buyer_opt}, when values are additive, the separate sales mechanism $\mM^{Sep}_{p^*_{n}}$ where $p^*_{n}=\left(\frac{\overline{p}^*_n}{n},\dots,\frac{\overline{p}^*_n}{n}\right)$ is also an optimal mechanism for the seller in response to $G^*_{n}$.

\medskip

Finally, we define
$$\text{CS}_n:= \tilde{\mu}-\frac{\overline{p}^*_n}{n},$$
to be the \textit{average consumer surplus}. This is the total consumer surplus from the buyer-optimal outcome divided by the number of goods $n$. The following result, the proof of which is in \cref{app:proofs}, shows that this average consumer surplus decreases in the number of goods.

\medskip

\begin{theorem}\label{thm:comp_statics}
	Suppose valuations are additive and each dimension of the type is iid. Then the average consumer surplus in the buyer-optimal outcome is decreasing in the number of goods, 
	$$CS_n \geq CS_{n+1}\quad \text{ for all } n\geq 1.$$
	Moreover, as $n$ grows large, the average consumer surplus in the buyer-optimal outcome converges to zero ($\lim_{n\to \infty} \text{CS}_n=0$) and the seller extracts all the surplus.  
\end{theorem} 

This result (proved in the appendix) is a straightforward consequence of \cref{thm:buyer_opt}. For intuition, consider the buyer-optimal outcome $\left(G^*_{n+1}, \mM^{Sep}_{p^*_{n+1}}\right)$ for the $n+1$ goods in which $G^*_{n+1}$ is perfectly correlated and the seller chooses a separate sales mechanism. Let $\widetilde{G}^*_{n}$ be the marginal distribution of $G^*_{n+1}$ over the first $n$ dimensions $(s_1,\dots,s_n)$. Note that $\widetilde{G}^*_{n}$ is a perfectly correlated signal for the type space $[\theta_\ell,\theta_h]^n$ distributed by $\tilde{F}\times\cdots \times \tilde{F}$. For the perfectly correlated signal $\widetilde{G}^*_{n}$, the same separate sales mechanism $\mM^{Sep}_{p^*_{n+1}}$ at prices $p^*_{n+1}$ is an optimal mechanism for the seller implying that buyer-optimal outcome must generate at least $CS_{n+1}$ in average consumer surplus.

\medskip

This result highlights the nuanced interplay between information and screening. In an additive values environment, there is always a buyer-optimal outcome in which the seller's optimal mechanism is separate sales. Thus, when we increase the number of goods, the seller still finds it optimal to sell each good separately or, in other words, does not strictly benefit from bundling together different goods. But yet the average consumer surplus decreases. This is because the information that the buyer receives changes as the number of goods increases. In order to prevent the seller from multi-dimensional screening, the signal in the buyer-optimal outcome must introduce correlation (by injecting noise) and does so by only providing information to the buyer about his value for the grand bundle. As the number of goods increase, such correlation surrenders more surplus to the seller until, in the limit, she can extract all the surplus.

\section{The optimal informationally robust mechanism}\label{sec:robust}

The buyer-optimal outcome characterizes the informational environment that is most advantageous for the buyer. Here, the timing is such that an information designer chooses the signal first, and then the seller best responds. It is equally natural to think of the alternative timing: the seller first picks her mechanism, and then nature chooses the signal in response. This has at least two interpretations. The first captures a seller who does not know the exact type distribution and is worried about model misspecification. The second interpretation is that the buyer acquires information after observing the mechanism, but the seller does not know the buyer's information acquisition technology. The seller evaluates each mechanism based on the worst-case profits taken with respect to all possible signal realizations and buyer best-responses.

\medskip

We begin with a few definitions.

\medskip

\noindent \textbf{Revenue Guarantee:} We say that a mechanism $\mM$ provides a \textit{revenue guarantee} of $\pi$ if
$$\Pi( G,\mM,\sigma)\geq \pi$$
for all signals $ G\in\mathcal{G}$ and buyer best responses $\sigma\in \Sigma(\mM)$.

\medskip

\noindent \textbf{Informationally Robust Mechanism:} Formally, we define the \textit{optimal informationally robust mechanism} $\umM^{\star}$ as the mechanism that solves
$$\umM^{\star}\in \argmax_{\mM}\; \inf_{ G\in \mathcal{G},\sigma\in\Sigma(\mM)}  \Pi( G,\mM,\sigma).$$
This is the mechanism we aim to characterize. In words, it provides the seller the highest revenue guarantee against all possible signals and best responses by the buyer. Note the difference with the (standard) definition of an optimal mechanism in \cref{sec:buyer_opt_out} where the buyer is assumed to break indifference in favor of the seller. As with the buyer-optimal outcome, we will explicitly construct the mechanism which will show that the maximum is obtained. 

\medskip

In fact, we will show that it takes the following simple form.

\medskip

\noindent \textbf{Random Pure Bundling Mechanism:} A \textit{random pure bundling mechanism} $\mM^{rPB}=(M^{rPB},q^{rPB},t^{rPB})$ has a message space $M^{rPB}=\overline{S}$ given by the set of possible posterior grand bundle estimates and an allocation rule 
\begin{equation}\label{eq:random_pb}
	q^{rPB}(m,b)=0 \text{ if } b\neq N \;\; \text{ and }\;\; q^{rPB}(m,N)+q^{rPB}(m,\emptyset)=1 \tag{rPB}
\end{equation}
for all $m\in M^{rPB}.$ In words, these are mechanisms in which the buyer is only ever allocated the grand bundle; however this allocation could be random. Put differently, the buyer is effectively offered a menu of prices and probabilities where each menu item corresponds to the buyer paying a price in exchange for receiving the grand bundle with the given probability.

\medskip

Before characterizing the optimal informationally robust mechanism, it is worth relating it to the buyer-optimal outcome. Observe that we must have 
\begin{equation}\label{eq:maxminmax}
	\max_{\mM}\; \inf_{ G\in \mathcal{G},\sigma\in\Sigma(\mM)}  \Pi( G,\mM,\sigma)\leq \max_{\mM}\; \inf_{\sigma\in\Sigma(\mM)}  \Pi( G^*,\mM,\sigma)=\pi^*=\min_{ G\in \mathcal{G}} \max_{\mM,\sigma\in\Sigma(\mM)}  \Pi( G,\mM,\sigma)
\end{equation}
where, recall that, $ G^*,\pi^*$ are respectively the signal we construct and the profit in the buyer-optimal outcome (\cref{thm:buyer_opt}). The first equality follows from the observation that $\mM^{PB}_{\pi^*}$ is an optimal mechanism in response to the signal $ G^*$ and that, for this signal and mechanism, all buyer best responses result in the same profit. In words, this shows that the seller's worst case profit from the optimal informationally robust mechanism must be lower than that from the buyer-optimal outcome. 

In the next result, we show that these are actually equal and, more importantly, we use this fact to characterize the optimal informationally robust mechanism. Note that we cannot immediately employ a minimax theorem to show the equality because the infimum on the left and the maximum on the right in \eqref{eq:maxminmax} are taken over both signals and buyer best responses. This fact is the point of departure for \citet{brooks2020strong} who study the relation between the max-min and min-max problems in general multi-agent environments (we discuss their paper below).

\begin{theorem}\label{thm:robust}
	There is a random pure bundling mechanism that is an optimal informationally robust mechanism. This mechanism provides a revenue guarantee of $\pi^*$, the seller's revenue in the buyer-optimal outcome.
\end{theorem}

\begin{proof}
	The proof adapts the argument of the proof of Proposition 1 in \citet{DuEcma2018}. We provide it here for completeness and flag the main step where we reduce our multidimensional problem to the single good case that he analyzes.

	We begin by stating a property of the signal  $ G^*$ from the buyer-optimal outcome we derived in the proof of \cref{thm:buyer_opt}. Recall that this signal is perfectly correlated and induces a truncated Pareto distribution $\overline{G}^*=H_{\pi^*}$ on grand bundle estimates. It is possible to show\footnote{See the proof of Proposition 1 in \citet{DuEcma2018} or Lemma 3 in \citet{RaRoSzWP2020}.} that there must be a $\os^*\in (\pi^*,\beta(\pi^*))$ in the interior of the support of $H_{\pi^*}$ such that
	\begin{equation}\label{eq:sosd_bind}
		\int_{n\theta_{\ell}}^{\os^*} \overline{F}(x) \td x=\int_{n\theta_{\ell}}^{\os^*} H_{\pi^*}(x) \td x.
	\end{equation}
	For some intuition, it is straightforward to show that 
	the function $I_\pi(y):=\int_{n\theta_{\ell}}^y H_\pi(x) \td x$ is continuous and decreasing in $\pi$ (for a given $y$). In other words, whenever $\pi<\pi'$, then $I_{\pi}(y)\geq I_{\pi'}(y)$ for all $y\in [n\theta_{\ell}, n\theta_h]$ with strict inequality on $(\pi, \beta(\pi))$. Moreover, 
	recall that $H_{\pi^*}\in\mathcal{\overline{G}}$ being a distribution over posterior grand bundle estimates (induced by a signal) implies that $\int_{n\theta_{\ell}}^{\os} \overline{F}(x) \td x-\int_{n\theta_{\ell}}^{\os} H_{\pi^*}(x) \td x \geq 0$ for all $\os\in (n\theta_{\ell},n\theta_h)$ or that $\overline{F}$ is a mean-preserving spread of $H_{\pi^*}$. If this inequality was always slack, it would be possible to construct another truncated Pareto distribution $H_{\pi^*-\varepsilon}\in\mathcal{\overline{G}}$ for sufficiently small $\varepsilon>0$ which would violate the optimality of $H_{\pi^*}$.
	
	\medskip
	
	Equation \eqref{eq:sosd_bind} has two immediate consequences. First, because $\int_{n\theta_{\ell}}^{\os} \overline{F}( x) \td x-\int_{n\theta_{\ell}}^{\os} H_{\pi^*}( x) \td x\geq 0$ for all $\os\in (n\theta_{\ell},n\theta_h)$, this function must be minimized at $\os=\os^*.$ Hence, from   the first-order condition, we obtain
	\begin{equation}\label{eq:cdfs}
		\overline{F}(\os^*) =H_{\pi^*}(\os^*).
	\end{equation}
	 Second, doing integration by parts on both sides of \eqref{eq:sosd_bind} and using \eqref{eq:cdfs}, we get
	\begin{equation}\label{eq:exp_value}
		\int_{n\theta_{\ell}}^{\os^*}  x\td\overline{F}( x)=\int_{n\theta_{\ell}}^{\os^*}  x\td H_{\pi^*}( x).
	\end{equation}
	
	\medskip
	
	We now use the revenue $\pi^*$ from the buyer optimal outcome and the value $\os^*$ defined in \eqref{eq:sosd_bind} to construct a random pure bundling mechanism. We first provide an intuitive implementation before formally describing the allocation and transfer. So suppose the seller only offers the grand bundle for sale but randomizes over the price. Specifically, suppose the price is randomly drawn from the interval $\overline{p}\in[\pi^*,\overline{s}^*]$ with the cumulative distribution $\mathcal{P}$ given by
	$$\mathcal{P}(\overline{p})=\left\{\begin{array}{cl} 1 & \text{ if } \overline{p}>\overline{s}^*, \\ \left[\log\left(\frac{\overline{s}^*}{\pi^*}\right)\right]^{-1} \log\left(\frac{\overline{p}}{\pi^*}\right) & \text{ if } \pi^*\leq \overline{p}\leq \overline{s}^*, \\ 0 & \text{ if } \overline{p}<\pi^*. \end{array}\right.$$

	Clearly, a type $s\in S$ purchases the grand bundle when the realized price $\overline{p}\leq \overline{s}$; since the distribution of prices has no atoms, the seller's revenue is unaffected by the decision that the buyer makes at the zero probability event $\overline{p}= \overline{s}$. Therefore, the allocation and transfer for this mechanism can be written in the form of a random pure bundling mechanism $\mM^{rPB}$ with
	\begin{equation}\label{eq:rpb_mech}
		\begin{split}
			& q^{rPB}(s,b)=0 \quad \text{if }b\neq N,\\
			& q^{rPB}(s,N)=\mathcal{P}(\os)=\left\{\begin{array}{cl} \left[\log\left(\frac{\overline{s}^*}{\pi^*}\right)\right]^{-1} \min\left\{ \log\left(\frac{\os}{\pi^*}\right)\,,\, \log\left(\frac{\os^*}{\pi^*}\right) \right\} &  \text{ if } \os > \pi^*, \\ 0 &  \text{ if } \overline{s} \leq \pi^*, \end{array}\right.  \\
			& t^{rPB}(s)=\int_{n\theta_{\ell}}^{\os} \overline{p} \td \mathcal{P} (\overline{p})= \left\{\begin{array}{cl} \left[\log\left(\frac{\overline{s}^*}{\pi^*}\right)\right]^{-1} \min\left\{ \os-\pi^*\,,\, \overline{s}^*-\pi^*\right\} &  \text{ if } \os > \pi^*, \\ 0 &  \text{ if } \os \leq \pi^*. \end{array}\right.
		\end{split}
	\end{equation}
	Note that faced with this mechanism, truth-telling or $\sigma(s)=s$ for all $s\in S$ is a best response, and all best responses give the seller the same revenue.
	
	\medskip
	
	Now observe that, for any signal $ G\in\mathcal{G}$, the seller's profit from the mechanism $\mM^{rPB}$ satisfies
	\begin{align}
		\int_\Theta t^{rPB}(s) \td  G(s) & =\left[\log\left(\frac{\overline{s}^*}{\pi^*}\right)\right]^{-1}\int_{\pi^*}^{n\theta_h} \left(\min\left\{ \os-\pi^*\,,\, \overline{s}^*-\pi^*\right\}\right) \td\overline{G}(\os) \nonumber \\
		& \geq \left[\log\left(\frac{\overline{s}^*}{\pi^*}\right)\right]^{-1}\int_{n\theta_{\ell}}^{n\theta_h} \left(\min\left\{ \os-\pi^*\,,\, \overline{s}^*-\pi^*\right\}\right) \td\overline{G}(\os) \nonumber \\
		& \geq \left[\log\left(\frac{\overline{s}^*}{\pi^*}\right)\right]^{-1}\int_{n\theta_{\ell}}^{n\theta_h} \left(\min\left\{ \os-\pi^*\,,\, \overline{s}^*-\pi^*\right\}\right) \td\overline{F}(\os). \label{eq:lb_exp}
	\end{align}
	The first inequality follows from the fact that $\min\left\{ \os-\pi^*\,,\, \overline{s}^*-\pi^*\right\}<0$ for $\os< \pi^*$. The second inequality is a consequence of two facts. First, $G$ is a signal and so (from part (i) of \cref{lemma:corr_signal}) $\overline{F}$ is a mean-preserving spread of the distribution $\overline{G}$ (induced by $G$ over grand bundle estimates). Second, the function being integrated is concave and hence the inequality follows from the fact that $\overline{F}\succsim \overline{G}$.

	Observe that \eqref{eq:lb_exp} implies 
	$$\inf_{ G\in \mathcal{G},\sigma\in\Sigma(\mM^{rPB})} \Pi( G,\mM^{rPB},\sigma) \geq \left[\log\left(\frac{\overline{s}^*}{\pi^*}\right)\right]^{-1}\int_{n\theta_{\ell}}^{n\theta_h} \left(\min\left\{ \os-\pi^*\,,\, \overline{s}^*-\pi^*\right\}\right) \td\overline{F}(\os)$$
	or, in words, that mechanism $\mM^{rPB}$ has a revenue guarantee given by the right side. This is the key insight that allows us to employ \citet{DuEcma2018}'s argument; it shows that it is possible to generate this revenue guarantee in our multidimensional environment by effectively using a one dimensional mechanism.

	\medskip
	
	Recall that $ G^*$ induces a truncated Pareto distribution $\overline{G}^*=H_{\pi^*}$ over grand bundle estimates and so the seller's worst-case profit from $\mM^{rPB}$ when the signal is $ G^*$ is equal to
	\begin{align*}
		\inf_{\sigma\in\Sigma(\mM^{rPB})} \Pi( G^*,\mM^{rPB},\sigma) =&\left[\log\left(\frac{\overline{s}^*}{\pi^*}\right)\right]^{-1}\int_{\pi^*}^{n\theta_h} \left(\min\left\{ \os-\pi^*\,,\, \overline{s}^*-\pi^*\right\}\right)  \td H_{\pi^*}(\os)\\
		=&\left[\log\left(\frac{\overline{s}^*}{\pi^*}\right)\right]^{-1} \int_{n\theta_{\ell}}^{\os^*} (\os-\pi^*) \td H_{\pi^*}(\os)+\left[\log\left(\frac{\overline{s}^*}{\pi^*}\right)\right]^{-1}(1-H_{\pi^*}(\os^*))(\overline{s}^*-\pi^*)\\
		=&\left[\log\left(\frac{\overline{s}^*}{\pi^*}\right)\right]^{-1} \int_{n\theta_{\ell}}^{n\theta_h} \left(\min\left\{ \os-\pi^*\,,\, \overline{s}^*-\pi^*\right\}\right) \td\overline{F}(\os),
	\end{align*}
	where the last equality follows from \eqref{eq:cdfs} and \eqref{eq:exp_value}. In other words, for the signal $ G^*$ from the buyer optimal outcome, the seller's profit precisely equals the revenue guarantee \eqref{eq:lb_exp}.
	
	Finally, note that the seller's revenue also satisfies
	\begin{align*}
		\inf_{\sigma\in\Sigma(\mM^{rPB})} \Pi( G^*,\mM^{rPB},\sigma) = & \left[\log\left(\frac{\overline{s}^*}{\pi^*}\right)\right]^{-1}\int_{\pi^*}^{n\theta_h} \left(\min\left\{ \os-\pi^*\,,\, \overline{s}^*-\pi^*\right\}\right)  \td H_{\pi^*}(\os)\\
		=&\left[\log\left(\frac{\overline{s}^*}{\pi^*}\right)\right]^{-1} \int_{\pi^*}^{\overline{s}^*} (\os-\pi^*)\frac{\pi^*}{\os^2} \td \os+ \left[\log\left(\frac{\overline{s}^*}{\pi^*}\right)\right]^{-1} (\overline{s}^*-\pi^*)\frac{\pi^*}{\os^*} \\
		=&\pi^*.
	\end{align*}

	But inequality \eqref{eq:maxminmax} shows that $\pi^*$ is the highest possible revenue guarantee a mechanism can provide and this implies that 
	$$\max_{\mM}\; \inf_{ G\in \mathcal{G},\sigma\in\Sigma(\mM)}  \Pi( G,\mM,\sigma)=\inf_{ G\in \mathcal{G},\sigma\in\Sigma(\mM^{rPB})} \Pi( G,\mM^{rPB},\sigma)= \inf_{\sigma\in\Sigma(\mM^{rPB})}  \Pi( G^*,\mM^{rPB},\sigma)=\pi^*$$
	and so $\mM^{rPB}$ is an optimal informationally robust mechanism and this completes the proof.
\end{proof}

\cref{thm:robust} provides one possible motivation for the ubiquity of pure bundling (especially for digital goods): this mechanism guarantees the best possible profit evaluated against model misspecification. Note that, unlike \cref{thm:buyer_opt}, even with additive values, separate sales cannot immediately be employed here as an alternative optimal informationally robust mechanism. To see this, suppose the seller offered a separate sales mechanism with a random but perfectly correlated price vector. Observe that there can be two separate signals $ G$ and $ G'$ that induce the same distribution over grand bundle estimates, but that generate different profits given this separate sales mechanism. Thus, we can no longer employ the simple argument that we did to generate the revenue guarantee given by \eqref{eq:lb_exp}. Loosely speaking, random pure bundling mechanisms have an advantage in providing higher revenue guarantees because they reduce the buyer's private information to a single dimension and, in doing so, the worst-case signal is effectively drawn from a smaller set.

\medskip

Finally, as we have already mentioned, the fact that the max-min and min-max problems have the same solution for the sale of a single good was first observed by \citet{DuEcma2018}. In independent and contemporaneous work, \citet{brooks2020strong} generalize this insight to a variety of different settings (including multiple goods with additive values) with multiple buyers who have interdependent values. They consider finite type spaces, so the fact that this equivalence arises in our model is not per se implied by any of their results. More substantively, however, the aims of our respective papers are different. Our goal is to derive qualitative properties of the seller's optimal mechanism in two different information environments; the fact that the seller gets the same revenue in both is a fact we use to prove \cref{thm:robust} not a main focus of this paper. Moreover, we show (in \cref{cor:seller_min}) that the seller's profit in every buyer-optimal outcome equals the value of the objective from the solution to the min-max problem (this equivalence is not a priori apparent). By contrast, \citet{brooks2020strong} aim to show the equivalence of the min-max and max-min problems very generally, but they do not derive the seller's optimal mechanism in either; instead, their results are meant to provide a means for efficient numerical simulation. In the multiple-agent environment they consider, this equivalence is harder to establish than in our single agent setting due to the possibility of multiple equilibria.

\section{Concluding remarks}\label{sec:conclusion}

In this paper, we study a general multidimensional screening problem for a seller in two different information environments with buyer learning. In the first, we derive the optimal mechanism under the information structure that maximizes consumer surplus. In the second, we derive the optimal informationally robust mechanism which provides the highest revenue guarantee for the seller against all possible information structures. We show that pure bundling emerges in both and that the seller's profit is the same in both problems. Our main theoretical insight is that the introduction of buyer learning allows us to reduce the seller's problem to a one-dimensional counterpart which, in turn, considerably simplifies this typically intractable problem.

We end with a brief discussion of the key assumption driving our results: the prior type distribution $F$ is exchangeable. As we discussed, this is a considerably more general environment than what is typically analyzed in the multidimensional screening literature, but it is worth pointing out that our proof strategy does not extend to general distributions.  This can happen even if $F$ is such that each $\theta_i$ has   identical marginal distributions; for instance, when the buyer's values for two goods are positive correlated but negatively correlated for a different pair.  Our proof of \cref{lemma:corr_signal}, which uses the exchangeability of the prior, does not generalize to this case.   Therefore, \cref{thm:buyer_opt} cannot be directly generalized, since it builds on this result that any distribution over posterior grand bundle estimates (induced by a signal) can be induced by a perfectly correlated signal.  
We view this richness to be yet another interesting feature of this multi-dimensional environment and, in future work, we hope to generalize our results to such asymmetric environments.

\newpage

\appendix

\section{Missing proofs from the text}\label{app:proofs}

\begin{proof}[Proof of \cref{cor:seller_min}]

Suppose, for contradiction, that there is an outcome $\left( G,\ (q,t)\right)$ in which the seller's profit $\pi$ is strictly lower than the profit $\pi^*$ from a buyer-optimal outcome $\left( G^*,\ (q^*,t^*)\right)$. Steps 1 and 2 in the proof of \cref{thm:buyer_opt} show that it is possible to find an outcome $\left( G',\ (q',t')\right)$ in which trade is efficient and the seller's profit from her optimal mechanism is exactly $\pi$. But this is a contradiction because the consumer surplus in outcome $\left( G',\ (q',t')\right)$ will be $\overline{\mu}-\pi>\overline{\mu}-\pi^*$ contradicting the optimality of $\left( G^*,\ (q^*,t^*)\right)$.

\end{proof}

\medskip

\begin{proof}[Proof of \cref{thm:comp_statics}]
	Denote the set of signals when there are $n$ goods by $\mathcal{G}_n$ and the set of distributions on grand bundle estimates that they induce by $\overline{\mathcal{G}}_n$. Let $\tilde{F}_n$ denote the distribution of the average value $\frac{\theta_1+\cdots+\theta_n}{n}\in [0,1]$. Then observe that any Pareto distributed grand bundle signal $H_{\alpha_n}\in \overline{\mathcal{G}}_{n}$ will satisfy $\tilde{F}_n\succsim H_{\frac{\alpha_n}{n},\frac{\beta(\alpha_n)}{n}}$ or, in words, that $H_{\frac{\ua_n}{n},\frac{\oa_n}{n}}$ is a \textit{signal for the average value of $n$ goods}. The converse is also true and $\tilde{F}_n\succsim H_{\frac{\alpha_n}{n},\frac{\beta(\alpha_n)}{n}}$ implies $H_{\alpha_n}\in \mathcal{\overline{G}}_{n}$. Therefore, if we have a signal for the average value, there will be a corresponding grand bundle signal which implies we can use \cref{lemma:corr_signal} to find a perfectly correlated signal for the type vector.
	
	Additionally, note that $\tilde{F}_n\succsim \tilde{F}_{n+1}$. Clearly $\tilde{F}_n$ and $\tilde{F}_{n+1}$ have the same mean $\tilde{\mu}$ and the former will have a higher variance because each $\theta_i$ is iid with distribution $\tilde{F}$. The fact that $\tilde{F}_n$ also second order stochastically dominates $\tilde{F}_{n+1}$ is a well known property of the distribution of sample means.
	
	From the proof of \cref{thm:buyer_opt}, we know that there are buyer-optimal outcomes $\left(G^*_{n},\, \mM^{Sep}_{p^*_{n}}\right)$ and $\left(G^*_{n+1},\, \mM^{Sep}_{p^*_{n+1}}\right)$ for $n$ and $n+1$ goods respectively where $G^*_{n}$ and $G^*_{n+1}$ are perfectly correlated and induce truncated Pareto distributed grand bundle estimates $H_{\overline{p}^*_{n}}\in \mathcal{\overline{G}}_{n}$ and $H_{\overline{p}^*_{n+1}}\in \mathcal{\overline{G}}_{n+1}$.
	
	We have argued above that
	\begin{equation}\label{eq:avg}
		\tilde{F}_n\succsim \tilde{F}_{n+1}\succsim H_{\frac{\overline{p}^*_{n+1}}{n+1},\frac{\beta(\overline{p}^*_{n+1})}{n+1}}.
	\end{equation}	
	In words, this implies that, in addition to being a signal for the average value of $n+1$ goods, $H_{\frac{\overline{p}^*_{n+1}}{n+1},\frac{\beta(\overline{p}^*_{n+1})}{n+1}}$ is also a signal for the average value of $n$ goods. The optimality of $\overline{p}^*_n$ then implies that $\frac{\overline{p}^*_{n+1}}{n+1}\geq \frac{\overline{p}^*_{n}}{n}$ or, equivalently, that $CS_n \geq CS_{n+1}$.
	
	\medskip
	
	Finally, as $n\to \infty$,  the weak law of large numbers implies that $\tilde{F}_n\overset{d}{\to} \delta_{\tilde{\mu}}$ where $\overset{d}{\to}$ denotes convergence in distributions and $\delta_{\tilde{\mu}}$ is the Dirac measure that assigns mass 1 to the point $\tilde{\mu}$. Therefore, we will also have $H_{\frac{\overline{p}^*_{n}}{n},\frac{\beta(\overline{p}^*_{n})}{n}}\overset{d}{\to} \delta_{\tilde{\mu}}$ and so,	in the limit, $\frac{\overline{p}^*_{n}}{n}\to \tilde{\mu}$ and the seller extracts all the surplus.
\end{proof}

\newpage

\bibliographystyle{econometrica}
\bibliography{literature}

\end{document}